\documentclass[11pt,onecolumn,journal]{IEEEtran}
\usepackage[mathscr]{eucal}
\usepackage[cmex10]{amsmath}
\usepackage{amssymb,amsmath,amsthm,amsfonts,latexsym}
\usepackage{amsmath,graphicx}
\usepackage{verbatim}
\usepackage{bm}
\usepackage{algorithmic, cite}
\usepackage{algorithm}
\usepackage{verbatim}
\usepackage{textcomp}
\usepackage{mathrsfs}
\usepackage{subcaption}
\usepackage{epsfig,color}
\usepackage{hyperref}
\usepackage{epstopdf}
\usepackage{setspace}
\usepackage{fancyhdr}

\catcode`~=11 \def\UrlSpecials{\do\~{\kern -.15em\lower .7ex\hbox{~}\kern .04em}} \catcode`~=13

\allowdisplaybreaks[1]

\newcommand{\calF}{\mathcal{F}}

\newcommand{\calR}{\mathcal{R}}
\newcommand{\calS}{\mathcal{S}}

\newcommand{\bA}{\mathbf{A}}

\newcommand{\bG}{\mathbf{G}}

\newcommand{\bI}{\mathbf{I}}

\newcommand{\bQ}{\mathbf{Q}}

\newcommand{\bU}{\mathbf{U}}

\newcommand{\bV}{\mathbf{V}}

\newcommand{\bbR}{\mathbb{R}}

\DeclareMathAlphabet{\mathbsf}{OT1}{cmss}{bx}{n}
\DeclareMathAlphabet{\mathssf}{OT1}{cmss}{m}{sl}

\newcommand{\rvE}{\mathsf{E}}

\newcommand{\rvP}{\mathsf{P}}

\DeclareSymbolFont{bsfletters}{OT1}{cmss}{bx}{n}
\DeclareSymbolFont{ssfletters}{OT1}{cmss}{m}{n}
\DeclareMathSymbol{\bsfGamma}{0}{bsfletters}{'000}
\DeclareMathSymbol{\ssfGamma}{0}{ssfletters}{'000}
\DeclareMathSymbol{\bsfDelta}{0}{bsfletters}{'001}
\DeclareMathSymbol{\ssfDelta}{0}{ssfletters}{'001}
\DeclareMathSymbol{\bsfTheta}{0}{bsfletters}{'002}
\DeclareMathSymbol{\ssfTheta}{0}{ssfletters}{'002}
\DeclareMathSymbol{\bsfLambda}{0}{bsfletters}{'003}
\DeclareMathSymbol{\ssfLambda}{0}{ssfletters}{'003}
\DeclareMathSymbol{\bsfXi}{0}{bsfletters}{'004}
\DeclareMathSymbol{\ssfXi}{0}{ssfletters}{'004}
\DeclareMathSymbol{\bsfPi}{0}{bsfletters}{'005}
\DeclareMathSymbol{\ssfPi}{0}{ssfletters}{'005}
\DeclareMathSymbol{\bsfSigma}{0}{bsfletters}{'006}
\DeclareMathSymbol{\ssfSigma}{0}{ssfletters}{'006}
\DeclareMathSymbol{\bsfUpsilon}{0}{bsfletters}{'007}
\DeclareMathSymbol{\ssfUpsilon}{0}{ssfletters}{'007}
\DeclareMathSymbol{\bsfPhi}{0}{bsfletters}{'010}
\DeclareMathSymbol{\ssfPhi}{0}{ssfletters}{'010}
\DeclareMathSymbol{\bsfPsi}{0}{bsfletters}{'011}
\DeclareMathSymbol{\ssfPsi}{0}{ssfletters}{'011}
\DeclareMathSymbol{\bsfOmega}{0}{bsfletters}{'012}
\DeclareMathSymbol{\ssfOmega}{0}{ssfletters}{'012}

\newcommand{\bSigma	}{\bm{\Sigma}}

\newcommand{\bPhi}{\bm{\Phi}}

\newtheorem{theorem}{Theorem}

\newtheorem{corollary}[theorem]{Corollary}

\newcommand{\qednew}{\nobreak \ifvmode \relax \else
      \ifdim\lastskip<1.5em \hskip-\lastskip
      \hskip1.5em plus0em minus0.5em \fi \nobreak
      \vrule height0.75em width0.5em depth0.25em\fi}

\pdfoutput=1 

\title{Change Detection with Compressive Measurements}
\author{George K. Atia, {\em Member, IEEE}
\thanks{This work was supported in part by NSF Grant CCF-1320547.}
\thanks {George K. Atia is with the Department of Electrical Engineering and Computer Science, University of Central Florida.  Email: george.atia@ucf.edu.} }
\begin{document}
\maketitle
\begin{abstract}
Quickest change point detection is concerned with the detection of statistical change(s) in sequences while minimizing the detection delay subject to false alarm constraints. In this paper, the problem of change point detection is studied when the decision maker only has access to \emph{compressive} measurements. First, an expression for the average detection delay of Shiryaev's procedure with compressive measurements is derived in the asymptotic regime where the probability of false alarm goes to zero. Second, the dependence of the delay on the compression ratio and the signal to noise ratio is explicitly quantified. The ratio of delays with and without compression is studied under various sensing matrix constructions, including Gaussian ensembles and random projections. For a target ratio of the delays after and before compression, a sufficient condition on the number of measurements required to meet this objective with prespecified probability is derived.
\end{abstract}
\begin{IEEEkeywords}
Quickest change detection, Compressive measurements, Concentration inequalities
\end{IEEEkeywords}

\section{Introduction}
Massive amounts of heterogeneous and multi-dimensional data are generated on a daily basis. An important focus of recent research to deal with this data deluge has been on efficient data collection, storage and acquisition without throwing away information. As such, the last decade has witnessed significant developments in the design of clever sampling solutions such as compressive sampling \cite{dono-ieeetit-2006,cande-icm-2006}. These are techniques that characterize situations where one can find solutions to under-determined linear equations. Compressive sampling is rooted in the fact that many observed
signals are sparse (or compressible) in some known basis (dictionary). This sparseness is exploited to reconstruct the entire signal from relatively few measurements through design of sampling
matrices with distance-preserving properties \cite{cande-romb-tao-pam-2006}. These ideas have been used in a number of applications such as improving storage in computer networks \cite{indy-procieee-2010}, acquisition time in MRIs \cite{lust-dono-paul-mrm-2007,lust-et-al-ieeespm-2008}, and reduced radiation dose in X-ray CT \cite{x-ray,li-luo-2011}.

The vast majority of the work in this area has focused on signal sampling for reconstruction. Nevertheless, there is a also a growing interest to develop signal processing techniques that work directly on compressive measurements when the goal of the inference task does not necessarily require the reconstruction of the signal. For such tasks, the relevant performance metric could be different from the mean square error of the signal estimate. For example, the problem of signal detection from compressive measurements was considered in \cite{dave-et-al-ieeestsp-2010,duar-et-al-icassp-2006}. It was shown that further compression gains are achievable since we do not wish to reconstruct the signal, but rather care about minimizing the probability of misclassification. The authors in \cite{farh-hugh-icassp-2014} considered the problem of recovery of principal components from compressive measurements. This approach was used to estimate the parameters of Gaussian Mixture models directly from compressed data.

In this paper we consider the problem of change detection from compressive measurements. Change detection aims to detect statistical changes in data while minimizing the detection delay subject to false alarm constraints \cite{bass-book-1993}. This problem arises in various applications including anomaly/intrusion detection, surveillance systems, and structural health monitoring \cite{tart-et-al-tsp-2006,SHM-book-2006}. An extensive body of work in sequential analysis has been devoted to understand the fundamental delay/false alarm tradeoff and various formulations such as the minimax \cite{Lorden,poll-anstat-1985,mous-anstat-1986} and the Bayesian forumlations \cite{shir-book-1978,tart-veer-thpa-2005,veer-tart-itw-2002} were considered. The performances of various change point algorithms were further analyzed in the asymptotic setting where the probability of false alarm goes to zero \cite{tart-veer-thpa-2005,lai-ieeetit-1998}. The main questions that this paper seeks to address is whether and how change detection can be achieved using compressive measurements and what the associated performance is. We derive bounds on the Average Detection Delay (ADD) of the Shiryaev's procedure with compressive measurements using different matrix constructions, including sensing matrices drawn from a Gaussian ensemble and random projections in the asymptotic regime of vanishing false alarm probability. Also, we explicitly quantify the dependence of the delay on the compression ratio and the signal to noise ratio.

The rest of the paper is organized as follows. In Section \ref{sec:background} we provide preliminary background for the change point detection problem and Shiryaev's procedure. The problem setup in presented in Section \ref{sec:setup}. In Section \ref{sec:ccd}, we study the problem of change detection from compressive measurements and derive upper and lower bounds on the average detection delay in the asymptotic setting of vanishing false alarm probability. In Section \ref{sec:ccd_sparse}, we focus on compressed change detection of sparse phenomena. Numerical and simulation results are presented in Section \ref{sec:results}. We conclude in Section \ref{sec:conc}.

\section{Background: Change Detection}
\label{sec:background}
Independent and identically distributed (i.i.d.) random variables (or vectors), $Y_1, Y_2, \ldots$, from a distribution $f_0$ are observed, and at some unknown point $\lambda$ in time, the observed sequence is still i.i.d. but with a new distribution $f_1$. In the Bayesian setting \cite{tart-veer-thpa-2005}, $\lambda$ is random with probability distribution, $\pi_k = \rvP(\lambda = k)$. Change point detection aims to design a stopping rule to declare the occurrence of the change.
A stopping time $\tau$ for an observed sequence $\{Y_n\}_{n\geq 1}$ is measurable if the event $\{\tau \leq n\}$ belongs to the sigma algebra $\calF_n= \sigma(Y_1, Y_2, \ldots, Y_n)$. As in \cite{tart-veer-thpa-2005}, we define the Average Detection Delay (ADD) and the Probability of False Alarm (PFA) as
\begin{align}
\mbox{ADD}(\tau) &= \rvE^{\pi}[(\tau-\lambda)^+|\tau\geq\lambda]
\label{eq:ADD}
\end{align}
\begin{align}
\mbox{PFA}(\tau)&= \rvP^{\pi}(\tau<\lambda) = \sum_{k = 1}^\infty \pi_k \rvP_k(\tau < k).
\label{eq:PFA}
\end{align}
$\rvP^{\pi}$ and $\rvE^{\pi}$ denote the average probability measure (w.r.t. $\pi$) and expectation, and $(\tau-\lambda)^+ = \max(\tau-\lambda,0)$.

Defining the two hypotheses, $H_0 : \lambda > n$ and $H_1: \lambda\leq n$, it is not hard to show that the Likelihood Ratio (LR), $\Lambda_n$, for these two hypotheses is
\begin{align}
\Lambda_n = \frac{\pi_0}{1-\pi_0} + \frac{1}{\rvP(\lambda>n)}\sum_{k=1}^n \pi_k \prod_{t=k}^n \frac{f_1(y_t)}{f_0(y_t)}.
\label{eq:LR}
\end{align}

The optimal change-point detection procedure aims to minimize ADD subject to a constraint on PFA. In particular, subject to the constraint $\mbox{PFA}\leq\alpha$, Shiryaev showed that it is optimal to stop at the first time $\nu_A$ that $\Lambda_n$ exceeds a threshold $A$ that depends on $\alpha$ \cite{shir-book-1978}. Hence,
\begin{align}
\nu_A = \inf\{n\geq 1: \Lambda_n \geq A\}
\label{eq:shiryaev}
\end{align}
where $A$ should be chosen to satisfy the false alarm constraint with equality, which may only be possible in special settings.
Setting $A=\frac{1-\alpha}{\alpha}$ guarantees that $\nu_A\in\Delta(\alpha) = \{\tau: \mbox{PFA}(\tau)\leq\alpha\}$. Subsequent work established that Shiryaev's procedure is asymptotically optimal when $\alpha\rightarrow 0$ for the aforementioned choice of the threshold $A$ \cite{tart-veer-thpa-2005}.

\section{Problem Setup}
\label{sec:setup}
In this paper, we are interested in studying the problem of change point detection in a setting where the decision maker only has access to compressive measurements. The objective is to study and quantify the effect of compression on the average detection delay, and to derive conditions under which change-point detection can be carried out efficiently with compressive measurements. Next, we introduce the problem setup.

Conditioned on a change at time $\lambda=k$, we assume that the observations before change follow the model
\begin{align}
y_t = \bPhi w_t, ~~t = 1,\ldots, k-1.
\label{eq:ObsBeforeChange}
\end{align}
After change, the observations follow the model
\begin{align}
y_t = \bPhi (s + w_t),~ t \geq k
\label{eq:ObsAfterChange}
\end{align}
where $w_t$ is i.i.d $N(0,\sigma^2 \bI_M)$. The sensing matrix $\bPhi$ is $M\times N$, with $M << N$, and $s$ is a known signal in $\calS\subseteq\bbR^N$. We assume that the decision maker does not have control over the sensing matrix $\bPhi$. We further assume that $\lambda$ is geometric with parameter $\rho$. Simplifying the expression of $\Lambda_n$ in (\ref{eq:LR}), we get
\begin{align}
\Lambda_n = \frac{\pi_0}{1-\pi_0} + \frac{1}{(1-\rho)^n}\sum_{k=1}^n \pi_k \exp(Z_n^k)
\end{align}
where
\begin{align}
Z_n^k&=\frac{1}{\sigma^2}\sum_{t=k}^n \left(y_t^T (\bPhi \bPhi^T)^{-1}\bPhi s - \frac{1}{2}s^T\bQ s\right)\nonumber
\end{align}
\begin{align}
=\frac{1}{\sigma^2}\sum_{t=k}^n \left(y_t^T (\bPhi \bPhi^T)^{-1}\bPhi s - \frac{1}{2}\|\bQ s\|_2^2\right).
\end{align}
The matrix, $\bQ={\bPhi^T}(\bPhi\bPhi^T)^{-1}{\bPhi}$, is the orthogonal projection matrix on $\calR(\bPhi)$, the row space of $\bPhi$. We also observe that the statistic $\Lambda_n$ obeys the recursion
\begin{align}
\Lambda_n \hspace{-0.75mm}= \frac{1}{1-\rho}(\Lambda_{n-1} \hspace{-0.75mm}+\hspace{-0.75mm} \rho)\exp\hspace{-0.75mm}\Big\{\frac{1}{\sigma^2}&\Big(y_n^T(\bPhi\bPhi^T)^{-1}{\bPhi}s
\hspace{-0.75mm}-\hspace{-0.75mm}\frac{1}{2}\|\bQ s\|_2^2\Big)\hspace{-0.75mm}\Big\}
\end{align}

\section{Compressed Change Detection}
\label{sec:ccd}
First, we would like to characterize the performance of Shiryaev's procedure in the compressive measurements setting. In this section, we consider random constructions of the matrix $\bPhi$ and derive upper and lower bounds on $\mbox{ADD}(\nu_A)$. Second, we characterize the ratio of the delays with and without compression as a function of the compression ratio. We state the following theorem.
\begin{theorem}
\label{thm:general_matrices}
Let $\bPhi$ be an $M\times N$ random matrix with rank $M$ and \emph{unit norm rows}. Then, for any $s\in\calS$, $\nu_A$ in (\ref{eq:shiryaev}) satisfies
\begin{align}
\mbox{ADD}_{\ell}\leq\mbox{ADD}(\nu_A)\leq ADD_u, ~\mbox{as}~\alpha\rightarrow 0,
\end{align}
with probability at least $1-2e^{-cM\delta^2}$, for some constant $c > 0$ and $\delta \in (0,1)$, where
\begin{align}
\mbox{ADD}_{\ell} &= \frac{|\log\alpha|}{\frac{1}{2\sigma^2}(1+\delta)\frac{M}{N}\|s\|_2^2 + |\log(1-\rho)|}(1+o(1))\nonumber\\
\mbox{ADD}_u &= \frac{|\log\alpha|}{\frac{1}{2\sigma^2}(1-\delta)\frac{M}{N}\|s\|_2^2 + |\log(1-\rho)|}(1+o(1)),
\label{eq:ADD_bounds}
\end{align}
and $o(1)\rightarrow 0$ as $\alpha\rightarrow 0$.
\end{theorem}

\begin{proof}
By the asymptotic optimality of Shiryaev's procedure \cite{tart-veer-thpa-2005} we know that
\begin{align}
\mbox{ADD}(\nu_A)\sim\frac{|\log\alpha|}{D(f_1,f_0) + |\log(1-\rho)|}(1+o(1)),
\end{align}
since $Z_n^k$ converges to $D(f_1,f_0)$, the KL-divergence between $f_1$ and $f_0$. Since $f_1$ is $N(\bPhi s, \sigma^2\bPhi\bPhi^T)$ and $f_0$ is $N(0, \sigma^2\bPhi\bPhi^T)$, then,
\begin{align}
\mbox{ADD}(\nu_A)\sim\frac{|\log\alpha|}{\frac{1}{2\sigma^2}\|\bQ s\|_2^2 + |\log(1-\rho)|}(1+o(1)).
\label{eq:ADD_Gaussian}
\end{align}
The matrix $\bPhi$ has full row rank. By the reduced form of the SVD decomposition, we can write $\bPhi = \bU\mathbf{\bSigma} \bV^T$, where $\bU, \bV$ and $\bSigma$ are unitary, orthonormal and diagonal matrices, respectively. The matrix $\check\bPhi = \bSigma^{-1}\bU^T\bPhi$ has the same row space of $\bPhi$ and has orthonormal rows. Hence,
\begin{align}
\|\bQ s\|_2 & \hspace{-0.75mm}\buildrel(a)\over= \|\check\bPhi^T(\check\bPhi\check\bPhi^T)^{-1}{\check\bPhi}s\|_2\buildrel (b)\over= \|\check\bPhi^T\check\bPhi s\|_2 \hspace{-0.75mm}= \|\check\bPhi s\|_2.
\end{align}
(a) follows since $\bPhi$ and $\check\bPhi$ have the same row space and (b) follows since $\check\bPhi$ has orthonormal rows, i.e., $\check\bPhi\check\bPhi^T = I$.
Since $\check\bPhi$ is a random orthogonal projection, then $\|\check\bPhi s\|_2$ satisfies
\begin{equation}
(1-\delta)\frac{M}{N}\|s\|_2^2\leq \|\check\bPhi s\|_2^2\leq(1+\delta)\frac{M}{N}\|s\|_2^2
\end{equation}
with probability at least $1-2e^{-cM\delta^2}$ \cite{dasg-gupt-rsta-2003,dave-et-al-ieeestsp-2010}. The result follows.
\end{proof}
The next theorem, characterizes $\mbox{ADD}(\nu_A)$ when $\bPhi$ is drawn from a Gaussian ensemble.
\begin{theorem}
\label{thm:gaussian_matrices}
If $\bPhi$ has i.i.d. Gaussian zero-mean entries and $\rvE[\bPhi^T\bPhi] = \bI$, then for any fixed $s\in\calS$, $\nu_A$ in (\ref{eq:shiryaev}) satisfies
\begin{align}
\mbox{ADD}_{\ell}\leq\mbox{ADD}(\nu_A)\leq ADD_u, ~\mbox{as}~\alpha\rightarrow 0,
\end{align}
with probability at least $1-2e^{-cM\delta^2}$, for some constant $c > 0$ and $\delta \in (0,1)$, where $\mbox{ADD}_{\ell}$ and $\mbox{ADD}_u$ are as defined in (\ref{eq:ADD_bounds}).
\end{theorem}

\begin{proof}
The matrix $\bPhi$ is Gaussian, thus it satisfies the subgaussian concentration inequality. In other words, for $\delta\in(0,1)$ and for any given $s\in\calS$
\begin{align}
(1-\delta)\|s\|_2^2\leq\|\bPhi s\|_2^2\leq (1+\delta)\|s\|_2^2,
\label{eq:subgaussian}
\end{align}
with probability greater than or equal to $1-2e^{-cM\delta^2}$ \cite{dasg-gupt-rsta-2003,dave-et-al-ieeestsp-2010}. The row space of $\bPhi$ has a uniformly distributed orientation. As such, $\|\bQ  s\|_2$ is distributed as $\|\bQ  s\|_2$ for a random orthogonal projection. Thus, $\sqrt{\frac{N}{M}}\bQ$ satisfies the subgaussian concentration inequality, i.e.,
\begin{equation}
(1-\delta)\frac{M}{N}\|s\|_2^2\leq \|\bQ s\|_2^2\leq(1+\delta)\frac{M}{N}\|s\|_2^2.
\end{equation}
Replacing in (\ref{eq:ADD_Gaussian}), Theorem \ref{thm:gaussian_matrices} follows.
\end{proof}

\subsection{Detection delay and compression ratio}
The result of Theorems \ref{thm:general_matrices} and \ref{thm:gaussian_matrices} quantify the effect of the compression ratio $\gamma = \frac{M}{N}$ and the $\mbox{SNR} = \frac{\|s\|^2}{\sigma^2}$ on the detection delay. In particular, with the aforementioned probability the delays ratio $r$ of the average detection delay with compression to that without compression for the settings of Theorems \ref{thm:general_matrices} and \ref{thm:gaussian_matrices} satisfies
\begin{align}
r_{\ell}\leq r\leq r_u
\label{eq:ratioADDs}
\end{align}
where,
\begin{align}
r_{\ell} &=\frac{\mbox{SNR} + 2|\log(1-\rho)|}{\gamma(1+\delta)\mbox{SNR} + 2|\log(1-\rho)|}(1+o(1)) \nonumber\\
r_u &= \frac{\mbox{SNR} + 2|\log(1-\rho)|}{\gamma(1-\delta)\mbox{SNR} + 2|\log(1-\rho)|}(1+o(1)).
\end{align}

\section{Compressed Change Detection of Sparse Phenomena}
\label{sec:ccd_sparse}
In this section, we particularly focus on the special case where $\calS$ is the set of sparse signals of order $K$, i.e.,
\begin{align}
\calS = \{s\in \bbR^N: \|s\|_0\leq K\}.
\label{eq:set_of_sparse}
\end{align}
The goal is to determine the number $M$ of measurements needed to achieve a target delay ratio $r\leq r_0$ with probability $\geq 1-\beta$. The result follows directly from the previous analysis and the concentration of random matrices. We can readily state the following theorem.
\begin{theorem}
\label{thm:target_delay_ratio}
Consider the setup of Theorem \ref{thm:gaussian_matrices}. For $\delta,\beta\in(0,1)$, and a number of measurements $M$ satisfying
\begin{align}
M\geq\max(M_1,M_2),
\end{align}
the delay ratio $r$ is such that $r\leq r_0$ with probability at least $1-\beta$, where
\begin{align}
M_1 &= 2\frac{K\log(\frac{42}{\delta})+\log(\frac{2}{\beta})}{c\delta^2}, ~ \mbox{and} \label{eq:M_for_RIP}\\
M_2 &= \frac{N}{r_0(1-\delta)}\left(1-\frac{2(r_0-1)}{\mbox{SNR}}|\log(1-\rho)|\right)\label{eq:M_for_ratio}
\end{align}
\end{theorem}
\begin{proof}
Given $\delta\in(0,1)$ and if $M\geq M_1$, then $\bPhi$ and $\sqrt{\frac{N}{M}}\bQ$ satisfy the subgaussian concentration inequality property with probability at least $1-\beta$ \cite{dave-et-al-ieeestsp-2010}. Replacing in the RHS of (\ref{eq:ratioADDs}), the target delay ratio $r_0$ is met if $M>M_2$. The result follows.
\end{proof}
The following theorem establishes a general result for the detection delay based on sensing matrices that satisfy the RIP property \cite{cande-icm-2006,cande-romb-tao-pam-2006}.
\begin{theorem}
\label{thm:RIP_matrices}
If the matrix $\bPhi$ satisfies the RIP property of order K and constant $\delta$ \cite{cande-romb-tao-pam-2006}, and $\bG=\bPhi^T\bPhi$ is the Gram matrix, then $ADD(\nu_A)$ satisfies
\begin{align}
\mbox{ADD}_{\ell}\leq\mbox{ADD}(\nu_A)\leq ADD_u, ~\mbox{as}~\alpha\rightarrow 0,
\end{align}
where,
\begin{align}
\mbox{ADD}_{\ell} &= \frac{|\log\alpha|}{\frac{\mbox{SNR}}{2\lambda_{\min}(\bG)}(1+\delta) + |\log(1-\rho)|}(1+o(1))\nonumber\\
\mbox{ADD}_u &= \frac{|\log\alpha|}{\frac{\mbox{SNR}}{2\lambda_{\max}(\bG)}(1-\delta) + |\log(1-\rho)|}(1+o(1)),
\end{align}
and $\lambda_{\min}(\bG)$ and $\lambda_{\max}(\bG)$ denote the minimum and maximum eigenvalues of $\bG$, respectively.
\end{theorem}
\begin{proof}
Let $J = \{i\in [N]: s_i\ne 0\}$, denote the support set of $s$, where $s_i$ is the $i$-th entry of the vector $s$ and $[N]:=\{1,\ldots,N\}$. For a matrix $\bA$, let $\bG(\bA, J)$ denote the Gram matrix $\bA_J^T\bA_J$, where $\bA_J$ denotes the submatrix of $\bA$ with columns indexed by the set $J$. Hence, 
\begin{align}
\|\bQ s\|_2^2 &= \|\check\bPhi s\|_2^2\nonumber\\
&\geq\lambda_{\min}(\bG(\check\bPhi,J))\|s\|_2^2\nonumber\\
&\geq\frac{\lambda_{\min}(\bG(\bPhi,J))\|s\|_2^2}{\lambda_{\max}(\bG(\bPhi,[N]))}.
\label{eq:eig_bound}
\end{align}
The equality follows as before from the common row space and the orthonormal rows property. The first inequality follows since $s\in\calS$ and $J$ is the support set, while the second inequality follows from the SVD of $\check\bPhi$.
Similarly, we can prove an upper bound to get that
\begin{align}
\frac{\lambda_{\min}(\bG(\bPhi,J))\|s\|_2^2}{\lambda_{\max}(\bG)}\leq \|\bQ s\|_2 \leq \frac{\lambda_{\max}(\bG(\bPhi,J))\|s\|_2^2}{\lambda_{\min}(\bG)}.
\end{align}
The matrix $\bPhi$ satisfies the RIP property of order $K$ and $|J|\leq K$, which is equivalent to the requirement
\[
\lambda_{\min}(\bG(\bPhi,J))\geq 1-\delta ~~\mbox{and}~~\lambda_{\max}(\bG(\bPhi,J))\leq 1+\delta.
\]
This completes the proof.
\end{proof}

\subsection{Compressed change detection via wireless channels}
In various sensing applications, the sensor measurements are transmitted to a central unit via a wireless channel for further processing. In such cases, the goal may be to detect a change based on measurements received at the central node. These measurements are the result of the convolution of the transmitted signals with the channel impulse
response and hence can be represented in matrix form with the matrix $\bPhi$ being a Toeplitz matrix. As a direct application of Theorem \ref{thm:RIP_matrices}, our next result establishes a bound on the detection delay of the Shiryaev procedure in such settings based on the known RIP properties of Toeplitz matrices \cite{haup-et-al-ieeetit-2010}.
\begin{corollary}
\label{thm:toeplitz}
If $\bPhi$ is is an $M\times N$ Toeplitz matrix, with all distinct entries $\Phi_i$ i.i.d., Gaussian with zero mean and $\rvE[\Phi_i^2] = 1/M$, then for any $\delta\in(0,1),~\exists$ constants $c_1, c_2 > 0$ such that
\begin{align}
\mbox{ADD}(\nu_A)\leq  \frac{|\log\alpha|}{\frac{\mbox{SNR}}{2}\frac{1-\delta}{1+\delta\frac{N}{K}} + |\log(1-\rho)|}(1+o(1))
\end{align}
with probability at least $1-e^{-c_1 M/K^2}$, when $M > c_2 K^2\log N$.
\end{corollary}
\begin{proof}
In \cite{haup-et-al-ieeetit-2010}, it was shown that sufficiently large random Toeplitz matrices satisfy the RIP property with high probability. Hence, by Theorem \ref{thm:RIP_matrices} we only need to upper bound the largest eigenvalue of the matrix $\bG$. Bounds based on Ger\v{s}gorin circle theorem \cite{varg-book-2004} were derived in \cite{haup-et-al-ieeetit-2010}. To upper bound the maximum eigenvalue of $\bG = \bG(\bPhi,[N])$, first note that
\begin{align}
\lambda_{\max}(\bG)&\leq \max_{i \in [N]} G_{i,i} + \max_{i\in [N]}\sum_{\substack{j=1 \\ j\ne i}}^N |G_{i,j}|\\
&\leq \max_{i \in [N]} G_{i,i} + (N-1)\max_{i,j\in [N]} |G_{i,j}|,
\label{eq:gersh_thm}
\end{align}
by Ger\v{s}gorin circle theorem. Using the concentration bounds in \cite{haup-et-al-ieeetit-2010} on the entries of the Gram matrix, we have that
\begin{align}
P\left\{|\max_i G_{i,i}-1|\geq \frac{\delta}{K}\right\}&\leq 2N\exp\left(\frac{-\delta_1 M}{K^2}\right)\nonumber\\
P\left\{\max_{i,j}|G_{i,j}|\geq\frac{\delta}{K}\right\}&\leq 2N^2 \exp\left(\frac{-\delta_2 M}{K^2}\right)
\label{eq:concent_eigs}
\end{align}
for some constants $\delta_1, \delta_2 > 0$.
Combining (\ref{eq:concent_eigs}) and (\ref{eq:gersh_thm})
\begin{align}
\lambda_{\max}(\bG)\leq 1+\delta\frac{N}{K}
\end{align}
with probability $\geq 1-e^{-c_1 M/K^2}$, when $M > c_2 K^2\log N$. Under this condition, $\bPhi$ was shown to satisfy the RIP property \cite{haup-et-al-ieeetit-2010} establishing the result of Theorem \ref{thm:toeplitz}.
\end{proof}

\section{Numerical and Simulation Results}
\label{sec:results}
\begin{figure}
  \centering
  \includegraphics[width = 0.5\textwidth]{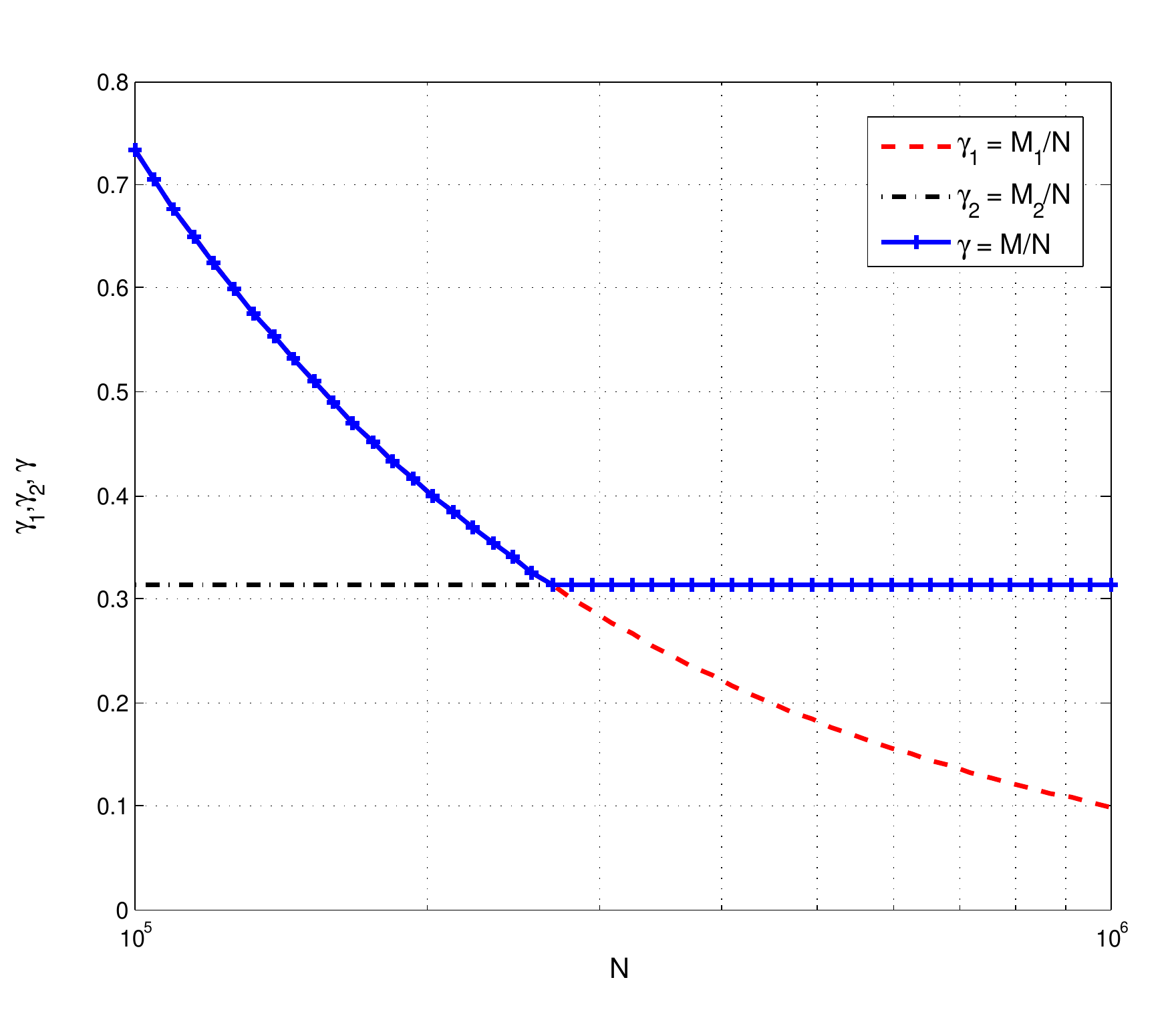}
  \vspace{-0.2cm}
  \caption{Compression ratio for a target delay ratio.}
  \label{fig:comp_ratio_vs_N}
\end{figure}
Theorem \ref{thm:target_delay_ratio} establishes a sufficient condition on the compression ratio $\gamma = M/N$ needed to achieve a target delay ratio $r_0$ with high probability. Fig. \ref{fig:comp_ratio_vs_N} shows the ratios $\gamma_1 = M_1/N$, $\gamma_2 = M_2/N$ and $\gamma = M/N$, as a function of the signal dimension $N$, for $r_0 = 4$, $\rho = 0.1$, SNR $=25$ dB, and $\beta = 0.1$. The size of the support $K$ is chosen to scale logartithmically with $N$. As shown, for large enough $N$, the curve corresponding to $\gamma_1$ (ensuring the subgaussian property), is dominated by the constant $\gamma_2$, which determines the compression gain. We note that the results hold irrespective of the choice of the signal $s\in\calS$ as we do not consider matching the matrix $\bPhi$ to $s$.

Fig. \ref{fig:ADD_PFA} displays the simulated tradeoff between the average delay of Shiryaev's procedure, $\mbox{ADD}(\nu_A)$, and the probability of false alarm, $\mbox{PFA}$, with compression, and the derived upper and lower bounds. The theoretical analysis of Theorem \ref{thm:gaussian_matrices} is shown to match the simulations. The results were obtained for $\rho = 0.1, N = 100, \delta = 0.5$ and $\mbox{SNR}$ = 5dB. $M$ was chosen to ensure that $\beta \leq 0.1$. Fig. \ref{fig:ADD_gamma} shows the average delay of Shiryaev's procedure as a function of the compression ratio $\gamma$, together with the theoretical upper and lower bounds.

\begin{figure}
        \centering
        \begin{subfigure}[b]{0.4\textwidth}
                \centering
                \includegraphics[width=\textwidth]{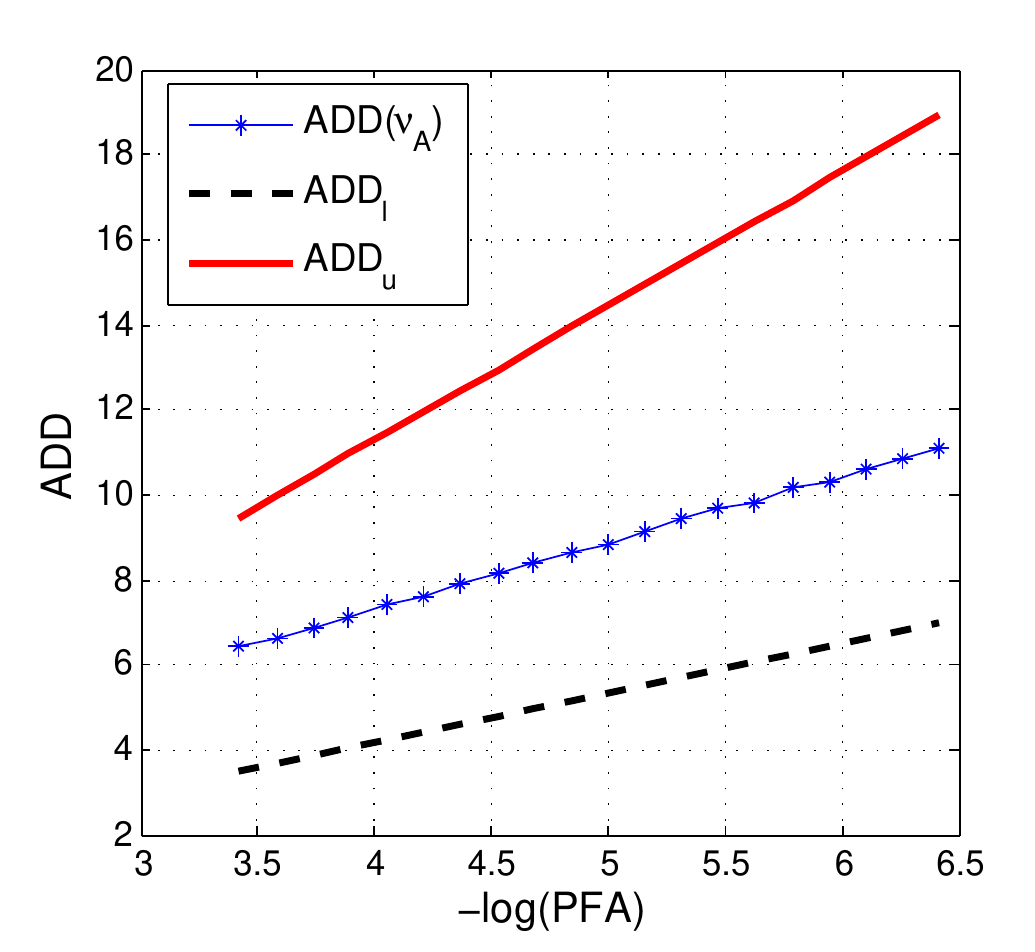}
                \caption{}
                \label{fig:ADD_PFA}
        \end{subfigure}%
        \begin{subfigure}[b]{0.405\textwidth}
                \centering
                \includegraphics[width=\textwidth]{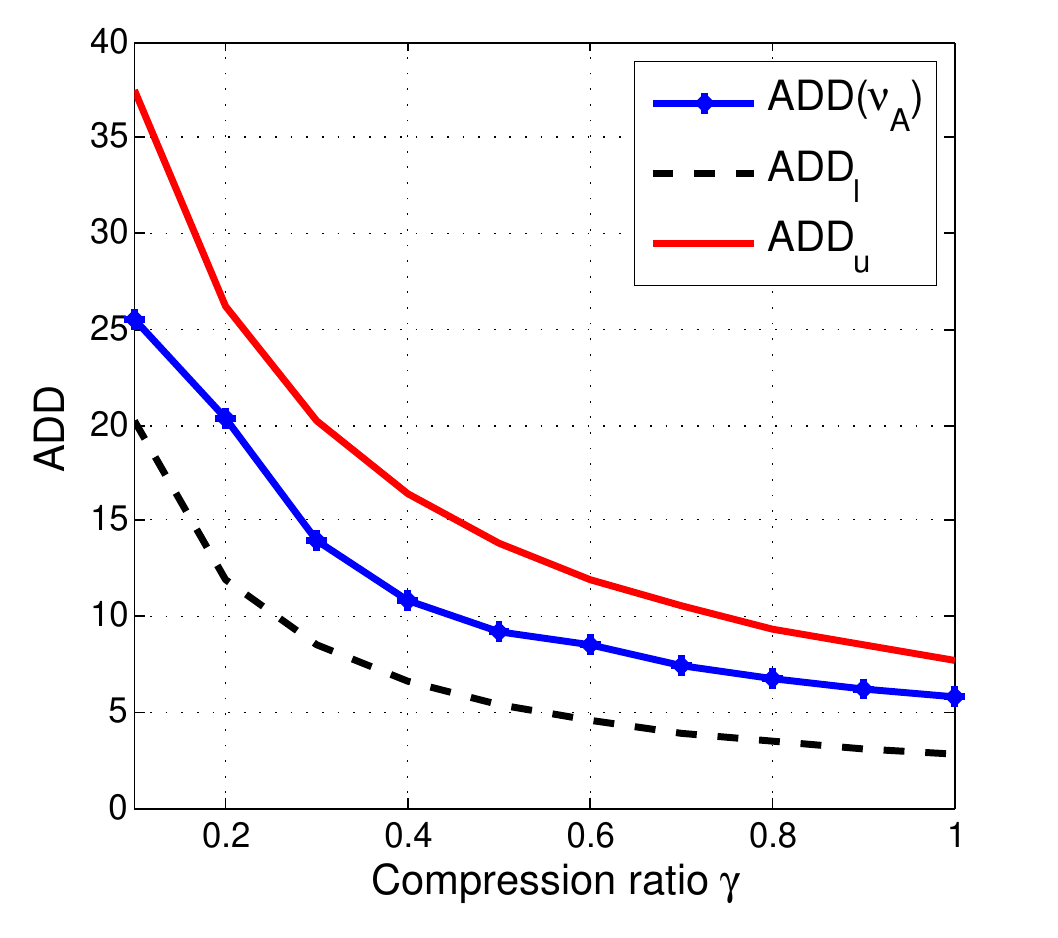}
                \caption{}
                \label{fig:ADD_gamma}
        \end{subfigure}
        \caption{(a) Average Detection Delay (ADD)-False alarm tradeoff of Shiryaev's procedure with compression and theoretical upper and lower bounds, (b) Average Detection Delay (ADD) vs the compression ratio $\gamma$ and theoretical bounds.}\label{fig:example2}
\end{figure}

\section{Conclusions}
\label{sec:conc}
We studied the problem of Bayesian change detection when the decision maker only has access to compressive measurements. We derived an expression for the average detection delay of Shiryaev's procedure with compressive measurements when the probability of false alarm is sufficiently small. We quantified the dependence of the delay on the compression ratio with various matrix constructions, including Gaussian ensembles and random projections, and derived upper and lower bounds on the average detection delay with compressive measurements. It was shown that the delay/false alarm tradeoff with compressive measurements depends on the projection on the row space of the sensing matrix, which admits a favorable concentration of measure for different sensing settings of interest.
\bibliographystyle{IEEEtran}
\bibliography{ccd_refs}


\end{document}